\documentclass[11pt, a4paper]{article}
\usepackage[margin=1in]{geometry}
\usepackage{amsmath}
\usepackage{amssymb}
\usepackage{amsthm}
\usepackage{graphicx}
\usepackage{caption}
\usepackage{flushend}
\newtheorem{theorem}{Theorem}
\newtheorem{definition}{Definition}
\newtheorem{lemma}{Lemma}
\usepackage[noadjust]{cite}
\usepackage[parfill]{parskip}

\title{Data-Driven Robust Control Using Reinforcement Learning}
\date{}
\usepackage{authblk}
\author[12]{Phuong D. Ngo \footnote{Corresponding author\\ Email: phuong.dinh.ngo@ehealthresearch.no (Phuong Ngo), fred.godtliebsen@uit.no (Fred Godtliebsen)}}
\author[2]{Fred Godtliebsen}
\affil[1]{Norwegian Centre for E-health Research, Troms\o, Norway}
\affil[2]{UiT The Arctic University of Norway, Troms\o, Norway}

\begin{document}
\maketitle

\begin{abstract}
This paper proposes a robust control design method using reinforcement-learning for controlling partially-unknown dynamical systems under uncertain conditions. The method extends the optimal reinforcement-learning algorithm with a new learning technique that is based on the robust control theory. By learning from the data, the algorithm proposed actions that guarantees the stability of the closed loop system within the uncertainties estimated from the data. Control policies are calculated by solving a set of linear matrix inequalities. The controller was evaluated using simulations on a blood glucose model for patients with type-1 diabetes. Simulation results show that the proposed methodology is capable of safely regulates the blood glucose within a healthy level under the influence of measurement and process noises. The controller has also significantly reduced the post-meal fluctuation of the blood glucose. A comparison between the proposed algorithm and the existing optimal reinforcement learning algorithm shows the improved robustness of the closed loop system using our method.
\end{abstract}

%

\section{Introduction}
Control of unknown dynamic systems with uncertainties is a challenge since most of the controllers require an exact mathematical model. Due to the fact that many processes are complicated, nonlinear and varying with time, a control algorithm that does not depend on a mathematical model and is able to adapt to time-varying conditions is required. A popular approach is to develop a universal approximator for predicting the output of unknown systems \cite{Lee2000}. Control algorithms can then be designed based on the parameters of the approximators. Many control techniques have been proposed based on this approach using neural networks and fuzzy logic. For example, Goyal et al. \cite{Goyal2015} proposed a robust sliding mode controller which can be designed from Chebyshev neural networks. Chadli and Guerra \cite{Chadli2012} introduced a robust static output feedback controller for Takagi Sugeno (TS) fuzzy models. Ngo and Shin \cite{Ngo} proposed a method to model unstructured uncertainties using type-2 fuzzy neural networks and the Takagi Sugeno fuzzy controller based on the model.

However, obtaining a good approximator requires a significant amount of data,  especially for a complicated model with many inputs and outputs. The data-driven model must also be updated frequently for time-varying systems. Also, many control design techniques assume uncertainties as functions of system parameters. However, in many cases, the causes of uncertainties are not known (i.e. unstructured uncertainties).

With the development of data science and machine learning, reinforcement learning (RL) has emerged as an effective method to control unknown nonlinear systems such as robotics manipulator or insulin-glucose models, see \cite{Bothe2013, DePaula2015a} and \cite{Ouyang2017}. Vrabie et al. \cite{Vrabie2013} proposed optimal adaptive control algorithms using RL for discrete and continuous dynamical systems. The principle of RL is based on the interaction between a decision-making agent and its environment \cite{Sutton2018}. In RL, the actor critic method is very popular and are used as the framework for many control algorithms. The critic agent uses current state information of the environment in oder to update the value or action value function. Then, the actor agent uses the value or action value function to calculate the optimal action.

However, many RL algorithms are lacking stability analysis of the control systems. Or the stability can only be ensured if the value function is estimated accurately. In many cases, this can not be achieved, especially at the beginning of the control process when the agent has just started interacting with the environment. Also in many applications, the state space is either continuous or very large, value function approximation must be used where approximation error can not be avoided. Therefore a new RL approach where stability can be guaranteed under uncertain conditions are very essential so that the algorithm can be used in many critical and safety-demanding systems such as aerospace or medical applications.

Type-1 diabetes is a disease caused by the lack of insulin secretion. The condition results in uncontrolled increase of blood glucose level if the patients are not provided with insulin doses. High blood glucose level can lead to both acute and chronic complications, and eventually result in failure of various organs. One of the major challenges in controlling the blood glucose is that the biochemical and physiologic kinetics of insulin and glucose is complicated, nonlinear, and only approximately known \cite{Wang2014}. Also, the stability of the control system is very essential in this case since unstable control effort will lead to life-threatening condition for the patients.

This paper proposes a data-driven robust control algorithm using reinforcement learning for partially-unknown dynamical systems. The purpose of the algorithm is to ensure the stability of the closed loop system under uncertainty conditions. The proposed methodology will be applied to a blood glucose model for testing its effectiveness in controlling the blood glucose level in patients with type-1 diabetes. 

The content of the paper is organized as follows. Section II describes the proposed robust reinforcement-learning algorithm. Section III shows the simulation results of the methodology. The conclusion is given in Section IV.

\section{Robust Control Using Reinforcement Learning}
In this paper, we consider a class of dynamical system which can be described by the following linear state-space equation:
\begin{equation}
\dot{x}(t) = Ax(t) + Bu(t)
\label{Eq:system_eq}
\end{equation}
where $x \in \mathbb{R}^{\mathit{n}}$ is the vector of $n$ state variables, $u \in \mathbb{R}^{\mathit{m}}$ is the vector of $m$ control inputs.  $A \in \mathbb{R}^ {\mathit{n\times n}} $ is the state matrix and $B \in \mathbb{R}^{\mathit{n\times m}} $ is the input matrix. It is assumed that matrix $A$ is unknown. Our target is to derive a control algorithm $u(t)$ that can regulate the state variables contained in $x(t)$ without knowing matrix $A$ and based on input and output data. 

As a RL framework, the proposed robust control algorithm consists of an agent that takes actions and learns the consequences of its actions in an unknown environment. The environment is defined by a state vector $x(t)$ that describes its states at time $t$. The action at time $t$ is represented by ${u}(t)$.
As a consequence of the action, a cost $r(t)$ is incurred and accumulated. The cost function $r(t)$ is assumed to be known and pre-defined as a function of the current state and action. The objective of the learning process is to minimize the total cost accumulation in the future. 

At each decision time point, the agent receives information about the state of the environment and chooses an action. The environment reacts to this action and transitions to a new state, which determines whether the agent receives a positive or negative reinforcement. Current RL techniques proposed optimal actions by minimizing the predicted cost accumulation. However, uncertainties due to noises in the data or inaccurate estimation of the cost accumulation can lead to sub-optimal actions and even unstable responses. Our target is to provide the agent with a robust and safe action that can guarantee the reduction of the future cost accumulation in the presence of uncertainties. The action calculated by the proposed algorithm may not be the optimal action that reduce the cost in the fastest way, but it can always guarantee the stability of the system, which is imperative in many critical applications.

\subsection{Estimation of the Value Function by the Critics}
In the RL context, the accumulation of cost over time, when starting in the state $x(t)$ and following policy $\pi$, is defined as the value function of policy $\pi$, i.e.:
\begin{equation}
	V^\pi(x(t))  =  E_{\pi} \left\{ \int_{t}^{\infty} \gamma^{\tau-t} r(\tau)d\tau \right\}
	\label{Eq:q_def}
\end{equation}
where $\gamma = 1$ is the discount factor. The cost $r(t)$ is assumed to be a quadratic function of the states:
\begin{equation}
r(t)   =x^T(t)Q x(t)
\label{Eq:cost_function}
\end{equation}
where the positive definite matrix $Q \in \mathbb{R}^ {\mathit{n\times n}}$ is symmetric and positive semidefinite (since the cost is assumed to be non-negative) contains the weighting factors of the variables that are minimized.

In order to facilitate the formulation of the stability condition in the form of linear matrix inequalities (LMI), the value function $V(x(t))$ is approximated by a quadratic function of the states:
\begin{equation}
	V^\pi(x(t))    \approx x^T(t)Px(t)
	\label{Eq:P_matrix}
\end{equation}
where the kernel matrix $P \in \mathbb{R}^ {\mathit{n\times n}}$ is symmetric and positive semidefinite (since matrix $Q$ in the cost function is symmetric and positive semidefinite). 

By using the Kronecker operation, the approximated value function can be expressed as a linear combination of the basis function $\phi(x(t)) = (x(t)\otimes x(t))$:
\begin{equation}
	\begin{aligned}
		V^\pi(x(t))  &\approx x^T(t)Px(t) = \text{vec}(P)^T(x(t)\otimes x(t))\\
		 &=  w^T(x(t)\otimes x(t)) = w^T \phi(x(t))
	\end{aligned}
\label{Eq:V_approx}
\end{equation}
where $w$ is the parameter vector, $\phi(x(t))$ is the vector of basis functions and $\otimes$ is the Kronecker product. The transformation between $w$ and $P$ can be done as follows:
\begin{equation}
w=vec(P) = [P_{11}, P_{21}, ..., P_{n1}, P_{12}, ..., P_{1n}, P_{nn}]^T
\end{equation}
where $P_{i,j}$ is the element of matrix $P$ in the $i^{\text{th}}$ row and $j^{\text{th}}$ column.
With $T$ as the interval time for data sampling, the integral reinforcement learning (IRL) Bellman equation can be used to update the value function \cite{Vrabie2013}:
\begin{equation}
V^\pi(x(t)) = \int_{t}^{t+T} \gamma^{\tau-t} r(\tau)d\tau + V^\pi(x(t+T))  
\end{equation}

By using the quadratic cost function (Eq. \ref{Eq:cost_function}) and the approximated value function (Eq. \ref{Eq:V_approx}), the IRL Bellman equation can be written as follows:
\begin{equation}
x^T(t)Px(t) = \int_{t}^{t+T} x(\tau)^TQ x(\tau)d\tau + x^T(t+T)Px(t+T)
\label{Eq:Bellman}
\end{equation}
or
\begin{equation}
	\begin{aligned}
		w^T\phi(x(t)) = \int_{t}^{t+T} &x(\tau)^TQ x(\tau)d\tau + w^T\phi(x(t+T)) 
	\end{aligned}
\end{equation}
At each iteration, $n$ samples along the state-trajectory are collected ($x^{(1)}(t), x^{(2)}(t),..., x^{(n)}(t)$). The mean value of $w$ can be obtained by using least square technique:
\begin{equation}
	\hat{w} = (XX^T)XY
	\label{Eq:LS}
\end{equation}
where
\begin{equation}
	X = [ \phi_\Delta^1 \quad  \phi_\Delta^2 \quad ... \quad \phi_\Delta^N]^T,
	\label{Eq:X}
\end{equation}
\begin{equation}
	\phi_\Delta^i = \phi(x^i(t)) - \phi(x^i(t+T))
\end{equation}
\begin{equation}
	Y = [ d(x^1(t)) \quad d(x^2(t)) \quad ... \quad d(x^n(t)) ]^T
	\label{Eq:Y}
\end{equation}
\begin{equation}
	d(x^i(t)) = \int_{t}^{t+T} x^i(\tau)^TQ x^i(\tau)d\tau
\end{equation}
with $ i = 1,2,...,N$. 

The confidence interval for the coefficient $w^{(j)}$  is given by:
\begin{equation}
w^{(j)} \in [\hat{w}^{(j)} -q_{1-\frac{\theta}{2}}\sqrt{\tau_j\hat{\sigma}^2}, \hat{w}^{(j)} +q_{1-\frac{\theta}{2}}\sqrt{\tau_j\hat{\sigma}^2}]
\label{Eq:LS}
\end{equation}
where $1-\theta$ is the confidence level, $q_{1-\frac{\theta}{2}}$ is the quantile function of standard normal distribution, $\tau_j$ is the jth element on the diagonal of $(XX^T)^{-1}$ and $\hat{\sigma}^2 = \frac{\hat{\epsilon}^T\hat{\epsilon}}{n-p}$, 
with $ \epsilon = Y - \hat{w}X$. From that, the uncertainty $\Delta w$ is defined as the deviation interval around the nominal value : 
\begin{equation}
\Delta w = \left[ -q_{1-\frac{\theta}{2}}\sqrt{\tau_j\hat{\sigma}^2},\  -q_{1-\frac{\theta}{2}}\sqrt{\tau_j\hat{\sigma}^2} \right]
\end{equation}
Matrices $\hat{P}$ and $\Delta P$ can be obtained by placing elements of $\hat{w}$ and $\Delta w$ into columns. 
\subsection{Policy Improvement by the Actor}
Linear feedback controller has been widely used as a stabilization tool for nonlinear systems where its dynamic behavior is considered approximately linear around the operating condition \cite{Kothare1996, Jyun-HorngFu1991, Eker2002}. Hence, in this paper, we use linear functions of the states with gain $K_i$ as the control policy at iteration $i$:
\begin{equation}
{u}(t)=\pi(x(t))=-K_ix(t)
\label{Eq:current_policy}
\end{equation}
and the level of uncertainty is constant during the controlling process, the task of the actor is to robustly improve the current policy such that the value function is guaranteed to be reduced during the next policy implementation. If the following differential inequality is satisfied:
\begin{equation}
\dot{V}_i (x(t)) + \alpha V_i(x(t)) \leq 0,
\end{equation}
 with some positive constant $\alpha$ then by using the comparison lemma (Lemma 3.4 in \cite{Khalil2002}), the derivative of function $\dot{V}_i (x(t))$ can be bounded by
\begin{equation}
\dot{V}_i (x(t)) \leq V_i (x(t_0)) e^{-\alpha(t-t0)}
\end{equation}
Therefore, maximizing the rate $\alpha$ will ensure a maximum exponential decrease in the value of $\dot{V}_i (x(t))$.

The following part shows the main results of the paper, which describe how the policy gain can improved during the learning process. Derivations of the results are provided in the stability analysis (Subsection \ref{sect:Stability}). In order to relax the stability condition and maximize the chance to obtain feasible solutions, we divide the problem into two cases depending on whether the policy update is made frequently or not. In the general policy-update case, it is assumed that the sign of all the state variables can be changed between each policy update interval. In the frequent policy-update case, it is assumed that the sign of all state variables cannot be changed between each policy update interval.

\subsubsection{General Policy Update} 
The improved policy $K_{i+1}$ can be obtained by solving the following linear program:

Minimize $\alpha$ subjected to the following LMIs
\begin{equation}
\begin{aligned}
\left[ \begin{array}{ccc}
U &  K_{i+1}^TB^T & \beta \\ 
B K_{i+1} & -\gamma_2 I & 0 \\ \beta & 0 & -\frac{\gamma_1}{\beta^2}I
\end{array} \right] \leq 0
\end{aligned}
\label{Eq:general_stable_cond}
\end{equation}
and
\begin{equation}
	\left[ \begin{array}{cc}
		\zeta 
		&  K^T_{i+1} \\ 
		K^T_{i+1} & -I
	\end{array} \right]
	\leq 0
\label{Eq:general_bound_cond}
\end{equation}
where the notion $S\leq0$ is a generalized inequality meaning $S$ is a negative semidefinite matrix, $\beta$ is the worst case norm of $\Delta P_i$, which can be estimated using the $\mu$ analysis \cite{Young1991},

\begin{equation}
\begin{aligned}
U = &  M - \hat{P}_i B K_{i+1} - K_{i+1}^TB^T\hat{P}_i  +  \alpha H  
\end{aligned}
\label{Eq:V}
\end{equation}

\begin{equation}
M =   -Q-K_{i}^TRK_{i} + \hat{P}_iBK_i + K_i^TB^T\hat{P}_i  + \gamma_1 K_i^TB^TB K_i  + \beta^2 \gamma_2 I
\label{Eq:M}
\end{equation}
and
\begin{equation}
H = \hat{P}_i + \frac{\beta^2}{2}I + I
\label{Eq:H}
\end{equation}
Inequality (\ref{Eq:general_stable_cond}) provides the stable condition and inequality (\ref{Eq:general_bound_cond}) provides the upper bound for the norm of the updated gain $K_{i+1}$. The derivation of (\ref{Eq:general_stable_cond}) is provided in Subsection \ref{sect:Stability}.
\subsubsection{Frequent Policy Update}
\begin{definition}
	Assume $A$ is a square matrix with dimension $n \times n$ and $x$ is a vector with dimension $n \times 1$. The \emph{maximize} operation on matrix $A$ and vector $x$ is defined as follows:
	\begin{equation}
	\mathrm{maximize}(A,x) = C
	\end{equation}
	where
	\begin{equation}
	C_{ij} = \left\{ \begin{array}{c}
	\mathrm{max}(A_{ij})\ \mathrm{if}\ x_ix_j \geq 0\\ 
	\mathrm{min}(A_{ij})\ \mathrm{if}\ x_ix_j < 0
	\end{array} \right. \mathrm{with} \text{ } i,j = 1..n
	\end{equation}
	\label{def:maximize}
\end{definition}
The improved policy $K_{i+1}$ can be obtained by minimizing $\alpha$ subject to
\begin{equation}
\begin{aligned}
\left[ \begin{array}{ccc}
V &  K_{i+1}^TB^T \\ 
B K_{i+1} & -\gamma_2 I
\end{array} \right] \leq 0
\end{aligned}
\label{Eq:frequent_stable_cond}
\end{equation}
and
\begin{equation}
\left[ \begin{array}{cc}
\zeta 
&  K^T_{i+1} \\ 
K^T_{i+1} & -I
\end{array} \right]
\leq 0
\label{Eq:frequent_bound_cond}
\end{equation}
where:
\begin{equation}
V =  M + \Delta P^T_{i,max} \Delta P_{i,max} \gamma_2 - \hat{P}_i B K_{i+1} \\
   - K_{i+1}^TB^T\hat{P}_i   + \alpha (\hat{P}_i + \frac{1}{2}\Delta P^T_{i,max} \Delta P_{i,max} + I) 
\label{Eq:V}
\end{equation}
and
\begin{equation}
\begin{aligned}
M = &  -Q-K_{i}^TRK_{i} + \hat{P}_iBK_i + K_i^TB^T\hat{P}_i  + H_i
\end{aligned}
\label{Eq:M}
\end{equation}
with
$\Delta P_{i,\text{max}} = \text{maximize}(\Delta P_{i}, x)$ and
$H_i = \text{maximize}(\Delta P_i B K_i   + K_i^TB^T\Delta P_i, x)$.

Similar to the general policy update, inequality (\ref{Eq:frequent_stable_cond}) provides the stable condition and inequality (\ref{Eq:frequent_bound_cond}) provides the upper bound for the updated gain $K_{i+1}$. The derivation of (\ref{Eq:frequent_stable_cond}) is provided in Subsection \ref{sect:Stability}.

\subsection{Stability Analysis} \label{sect:Stability}
With the control policy as described in Eq. \ref{Eq:current_policy}, the equation for the closed loop system can be derived as follows:
\begin{equation}
	\dot{x}(t) = Ax(t) - BKx(t) = (A-BK)x(t)
	\label{Eq:CL_eq}
\end{equation}

\begin{lemma} 
	Assume that the closed loop system described by Eq. \eqref{Eq:CL_eq} is stable, solving for $P$ in Eq. \eqref{Eq:Bellman} is equivalent to finding the solution of the underlying Lyapunov equation \cite{Vrabie2013}:
	\begin{equation}
	P(A-BK) + (A-BK)^TP= -Q
	\label{Eq:Lyapunov}
	\end{equation}
\end{lemma}

\begin{proof}
	We start with Eq. \ref{Eq:Lyapunov} and try to prove that matrix $P$ is also the solution of Eq. \eqref{Eq:Bellman}.
	Consider $V(x(t))   = x^T(t)Px(t)$, where $P$ is the solution of Eq. \eqref{Eq:Lyapunov}, we have:
	\begin{equation}
	\begin{aligned}
	\dot{V}(x(t)) &= \frac{d(x^T(t)Px(t))}{dt} \\
	&=\dot{x}^T(t)Px(t) + x^T(t)P\dot{x}(t) \\
	&= x^T(t)\left[(A-BK)^TP+ P(A-BK)\right]x(t) \\
	&=  -x^T(t)Qx(t) \quad \text{(using Eq. 27)}
	\label{Eq:Lemma_prove}
	\end{aligned}
	\end{equation}	
	Since the closed-loop system is stable, the Lyapunov equation \eqref{Eq:Lyapunov} has a unique solution, $P_i > 0$. From \eqref{Eq:Lemma_prove}, this solution will satisfy:
	\begin{equation}
	\frac{d(x^T(t)P_ix(t))}{dt}  =  -x^T(t)Qx(t)
	\end{equation}
	which is equivalent to 
	\begin{equation}
	x^T(t+T)Px(t+T) - x^T(t)Px(t)  = 
	\int_{t}^{t+T} -x^T(\tau)Qx(\tau)d\tau
	\end{equation}
	Therefore, $P$ is also the solution of Eq. \eqref{Eq:Bellman}.
\end{proof}

\begin{lemma} 
	Given matrices $E$ and $F$ with appropriate dimensions, the following linear matrix inequality (LMI) can be obtained:
	\begin{equation}
	EF^T + FE^T \leq EE^T + FF^T
	\end{equation}
	\label{lemma_matrix_product}
\end{lemma} 
\vspace{-0.4cm}
\begin{proof}
	From the properties of matrix norm, we have:
	\begin{equation}
	(E-F)(E-F)^T \geq 0
	\end{equation}
	which is equivalent to:
	\begin{equation}
	EE^T + FF^T - EF^T - FE^T \geq 0
	\end{equation}
	or
	\begin{equation}
	EF^T +FE^T \leq EE^T + FF^T
	\end{equation}
\end{proof}
\begin{theorem}	
	Consider a dynamic system that can be represented by Eq. \eqref{Eq:system_eq} with unknown state matrix. The estimated value function at iteration $i$ is $V_i(x(t))   = x^T(t)P_ix(t)$ with $P_i = \hat{P}_i+ \Delta P_i$. If:
	\begin{itemize}
		\item the current control policy ${u}(t)=\pi_i(x(t))=-K_ix(t)$ is stabilizing,
		\item the LMI given in (\ref{Eq:general_stable_cond}) is satisfied with some positive constants $\gamma_1$ and $\gamma_2$,
	\end{itemize}
then the closed loop system with the control policy ${u}(t)=-K_{i+1}x(t)$ is quadratic stable with convergence rate $\alpha$.
	\label{Theorem:Policy_Update}
\end{theorem}

\begin{proof}
	Since the current control policy is stable, the estimated parameter matrix $P_i$ is positive definite. Hence, $V_i(x(t)) = x_t^T P_i x_t > 0$. Here, we will use  $V_i(x(t))$ as the Lyapunov function for the updated control policy ${u}(t)=\pi_{i+1}(x(t))=-K_{i+1}x(t)$. For notation convenience, the state vector $x(t)$ and input vector $u(t)$ are denoted as $x_t$ and $u_t$, respectively. 
	We have:
	\begin{equation*}
	\begin{aligned}
	\dot{V}_i&(x(t)) + \alpha V_i(x(t)) \\
	= & \dot{x}_t^T P_ix_t + x_t  P_i\dot{x}_t^T + \alpha x_t^T P_ix_t\\
	= & (Ax_t+Bu_t)^T P_ix_t + x_t  P_i(Ax_t+Bu_t)^T + \alpha x_t^T P_ix_t\\
	= & x_t^T[P_i(A-BK_{i+1}) + (A-BK_{i+1})^T P_i + \alpha P_i]x_t \\
	= &x_t^T[P_i(A-BK_{i}) + (A-BK_{i})^TP_i + \alpha P_i]x_t
	+ x_t^T[P_iB(K_{i}-K_{i+1}) + (K_i-K_{i+1})^TB^T P_i  \alpha P_i]x_t\\
	= & -x_t^T[Q+K_{i}^TRK_{i}]x_t 
	 x_t^T[(\bar{P}_i+ \Delta P_i) B(K_{i}-K_{i+1})
	+  (K_i-K_{i+1})^TB^T (\bar{P}_i+ \Delta P_i) + \alpha \bar{P}_i +\alpha \Delta P_i]x_t\\
	= & x_t^T[-Q-K_{i}^TRK_{i} + \bar{P}_iBK_i + K_i^TB^T\bar{P}_i  +  \alpha \bar{P}_i 
	+  \Delta P_i B K_i  + K_i^TB^T\Delta P_i  - \Delta P_i B K_{i+1}   \\
	-& K_{i+1}^TB^T\Delta P_i   - \bar{P}_i B K_{i+1}  - K_{i+1}^TB^T\bar{P}_i +\alpha \Delta P_i ]x_t \\
	\end{aligned}
	\label{eq:the1:Vdot1}
	\end{equation*}
	
	By using Lemma \ref{lemma_matrix_product}, the following inequalities can be obtained:
	\begin{equation}
	\begin{aligned}
	\Delta P_i B K_i  & + K_i^TB^T\Delta P_i  
	\leq \frac{1}{\gamma_1} \Delta P_i \Delta P_i^T + \gamma_1 (B K_i)^T(B K_i) 
	 \leq \frac{\beta^2}{\gamma_1}  I + \gamma_1 K_i^TB^TB K_i 
	\end{aligned}
	\label{ieq:the1:prev_K}
	\end{equation}
	\begin{equation}
	\begin{aligned}
	- \Delta P_i B & K_{i+1}   - K_{i+1}^TB^T\Delta P_i  \leq \gamma_2\Delta P_i \Delta P_i^T + \frac{1}{\gamma_2} (B K_{i+1})^T(B K_{i+1})  
	 \leq \beta^2 \gamma_2  I + \frac{1}{\gamma_2} K_{i+1}^TB^TB K_{i+1}
	\end{aligned}
	\label{ieq:the1:next_K}
	\end{equation}
	\begin{equation}
	\begin{aligned}
	\alpha \Delta P_i   \leq \alpha \left( \frac{1}{2}\Delta P_i \Delta P_i^T + I \right) \leq \left( \frac{\alpha \beta^2}{2}+ \alpha \right) I
	\end{aligned}
	\label{ieq:the1:next_K}
	\end{equation}
	
	Hence, $\dot{V}_i(x(t)) + \alpha V_i(x(t))$ can be bounded by:
	\begin{equation*}
	\begin{aligned}
	\dot{V}_i(x(t)) & + \alpha V_i(x(t)) \leq x_t^T [ -Q-K_{i}^TRK_{i} + \bar{P}_iBK_i + K_i^TB^T\bar{P}_i  +  \alpha  \left(\bar{P}_i + \frac{\beta^2}{2}I + I\right)  \\
	&  + \beta^2 \gamma_2 I  + \frac{\beta^2}{\gamma_1}  I  
	 - \bar{P}_i B K_{i+1}  - K_{i+1}^TB^T\bar{P}_i + \gamma_1 K_i^TB^TB K_i  +  \frac{1}{\gamma_2} K_{i+1}^TB^TB K_{i+1}] x_t
	\end{aligned}
	\end{equation*}
	
	Using the Lyapunov theory, the system will be quadratic stable with the convergent rate $\alpha$ if $\dot{V}_i(x(t)) \leq -\alpha V_i(x(t))$. This condition will satisfy if:
	\begin{equation*}
	\begin{aligned}
	&x_t^T [ -Q-K_{i}^TRK_{i} + P_iBK_i + K_i^TB^TP_i  + \alpha  \left(P_i + \frac{\beta^2}{2}I + I\right)  + \beta^2 \gamma_2 I  + \frac{\beta^2}{\gamma_1}  I + \gamma_1 K_i^TB^TB K_i \\
	& + \frac{1}{\gamma_2} K_{i+1}^TB^TB K_{i+1}- P_i B K_{i+1}  - K_{i+1}^TB^TP_i  ] x_t \leq 0
	\end{aligned}
	\end{equation*}
	The above condition can be written in the matrix form as shown in the theorem.
\end{proof}

\begin{lemma} 
	Given $A$ as a square matrix with dimension $n \times n$ and $x$ as a vector with dimension $n \times 1$, the following LMI can be obtained:
	\begin{equation}
	x^TAx \leq x^TCx
	\end{equation} 
	where $C = \mathrm{maximize}(A,x)$ as in Definition \ref{def:maximize}.
	\label{lemma_matrix_maximize}
\end{lemma} 

\begin{proof}
	We have:
	\begin{equation}
		\begin{aligned}
			x^TAx &= \sum_{i,j=1,2..n} a_{ij} x_{i} x_{j} \leq \sum_{i,j=1,2..n} |a_{ij} x_{i} x_{j}|\\
			&=\sum_{i,j=1,2..n} c_{ij} x_{i} x_{j} = x^TCx
		\end{aligned}
	\end{equation}
	where $c_{ij} = \left\{ \begin{array}{c}
	\text{max}(a_{ij})\ \mathrm{if}\ x_ix_j \geq 0\\ 
	\text{min}(a_{ij})\ \mathrm{if}\ x_ix_j < 0
	\end{array} \right.  $ with $i,j = 1..n$
\end{proof}

\begin{theorem}	
	Consider a dynamic system that can be represented by Eq. \eqref{Eq:system_eq} with unknown state matrix. The estimated value function at iteration $i$ is $V_i(x(t))   = x^T(t)P_ix(t)$ with $P_i = \hat{P}_i+ \Delta P_i$. If:
	\begin{itemize}
		\item the current control policy ${u}(t)=\pi_i(x(t))=-K_ix(t)$ is stabilizing,
		\item the LMI given in (\ref{Eq:frequent_stable_cond}) is satisfied with some positive constant $\gamma_2$,
	\end{itemize}
then the closed loop system with the control policy ${u}(t)=-K_{i+1}x(t)$ is quadratic stable with convergence rate $\alpha$.
	\label{Theorem:Policy_Update}
\end{theorem}

\begin{proof}
	Since the current control policy is stable, the estimated parameter matrix $P_i$ is positive definite. Hence, $V_i(x(t)) = x_t^T P_i x_t > 0$. Here, $V_i(x(t))$ is used as the Lyapunov function for the updated control policy ${u}(t)=\pi_{i+1}(x(t))=-K_{i+1}x(t)$. For notation convenience, the state vector $x(t)$ and input vector $u(t)$ are denoted as $x_t$ and $u_t$, respectively. 
	We have:
	\begin{equation*}
	\begin{aligned}
	\dot{V}_i&(x(t)) + \alpha V_i(x(t)) \\
	= & \dot{x}_t^T P_ix_t + x_t  P_i\dot{x}_t^T + \alpha x_t^T P_ix_t\\
	= & (Ax_t+Bu_t)^T P_ix_t + x_t  P_i(Ax_t+Bu_t)^T + \alpha x_t^T P_ix_t\\
	= & x_t^T[P_i(A-BK_{i+1}) + (A-BK_{i+1})^T P_i + \alpha P_i]x_t \\
	= &x_t^T[P_i(A-BK_{i}) + (A-BK_{i})^TP_i + \alpha P_i]x_t 
	+   x_t^T[P_iB(K_{i}-K_{i+1}) + (K_i-K_{i+1})^TB^T P_i 
	+  \alpha P_i]x_t\\
	= & -x_t^T[Q+K_{i}^TRK_{i}]x_t 
	+  x_t^T[(\bar{P}_i+ \Delta P_i) B(K_{i}-K_{i+1})
	+  (K_i-K_{i+1})^TB^T (\bar{P}_i+ \Delta P_i) + \alpha \bar{P}_i +\alpha \Delta P_i]x_t\\
	= & x_t^T[-Q-K_{i}^TRK_{i} + \bar{P}_iBK_i + K_i^TB^T\bar{P}_i  +  \alpha \bar{P}_i 
	+  \Delta P_i B K_i  + K_i^TB^T\Delta P_i  - \Delta P_i B K_{i+1}   \\
	-& K_{i+1}^TB^T\Delta P_i   - \bar{P}_i B K_{i+1}  - K_{i+1}^TB^T\bar{P}_i +\alpha \Delta P_i ]x_t \\
	\end{aligned}	
	\end{equation*}
	
We have the following inequalities:
	\begin{equation}
	\Delta P_i B K_i   + K_i^TB^T\Delta P_i  \leq H_i
	\label{ieq:the1:prev_K}
	\end{equation}
	and
	\begin{equation}
	\begin{aligned}
	- \Delta P_i B  K_{i+1}   - K_{i+1}^TB^T\Delta P_i  
	&\leq \gamma_2\Delta P_i \Delta P_i^T + \frac{1}{\gamma_2} (B K_{i+1})^T(B K_{i+1})  \\
	& \leq \gamma_2\Delta P_{i,\text{max}} \Delta P_{i,\text{max}}^T + \frac{1}{\gamma_2} K_{i+1}^TB^TB K_{i+1} 
	\end{aligned}
	\label{ieq:the1:next_K}
	\end{equation}
	\begin{equation}
	\alpha \Delta P_i   \leq \alpha \left( \frac{1}{2}\Delta P_i \Delta P_i^T + I \right) \leq \alpha \left( \frac{1}{2}\Delta P_{i,\text{max}} \Delta P_{i,\text{max}}^T + I \right)
	\label{ieq:the1:next_K}
	\end{equation}
where $H_i = \text{maximize}(\Delta P_i B K_i   + K_i^TB^T\Delta P_i, x)$ \\and $\Delta P_{i,\text{max}} = \text{maximize}(\Delta P_{i}, x)$. The maximize operator is defined in Definition \ref{def:maximize}.

	Hence, $\dot{V}_i(x(t)) + \alpha V_i(x(t))$ can be bounded by:
	\begin{equation}
	\begin{aligned}
	\dot{V}_i(x&(t)) + \alpha V_i(x(t)) 
	\leq  x_t^T [ -Q-K_{i}^TRK_{i} + \bar{P}_iBK_i + K_i^TB^T\bar{P}_i  +  \alpha  \left(\bar{P}_i + \frac{1}{2}\Delta P_{i,\text{max}} \Delta P_{i,\text{max}}^T \right)  \\  
	& - \bar{P}_i B K_{i+1} - K_{i+1}^TB^T\bar{P}_i +  \gamma_2\Delta P_{i,\text{max}} \Delta P_{i,\text{max}}^T   + \frac{1}{\gamma_2} K_{i+1}^TB^TB K_{i+1} ] x_t 
	\end{aligned}
	\end{equation}
	Using the Lyapunov theory, the system will be quadratic stable with the convergent rate $\alpha$ if $\dot{V}_i(x(t)) \leq -\alpha V_i(x(t))$. This condition will satisfy if:
	\begin{equation*}
	\begin{aligned}
	&x_t^T [ -Q-K_{i}^TRK_{i} + \bar{P}_iBK_i + K_i^TB^T\bar{P}_i  +  \alpha  \left(\bar{P}_i + \frac{1}{2}\Delta P_{i,\text{max}} \Delta P_{i,\text{max}}^T \right) - \bar{P}_i B K_{i+1} - K_{i+1}^TB^T\bar{P}_i \\
	& +  \gamma_2\Delta P_{i,\text{max}} \Delta P_{i,\text{max}}^T  + \frac{1}{\gamma_2} K_{i+1}^TB^TB K_{i+1}  ] x_t \leq 0
	\end{aligned}
	\end{equation*}
	The above condition can be written in the matrix form as shown in the theorem.
\end{proof}

By using theorem 1 and 2, it can be seen that with the proposed improved policy, the closed loop system will be asymptotically stable.

\subsection{Robust Reinforcement Learning Algorithm}
The robust reinforcement learning algorithm for controlling partially unknown dynamically systems includes the following steps:

\subsubsection*{Initialization (step $i=0$)}
\begin{itemize}
	\item Select an initial insulin policy $u(t)=K_0 x(t)$.
\end{itemize}
\subsubsection*{Estimation of the value function (step $i=1,2,...$)}
\begin{itemize}
	\item Apply the control action $u(t)$ based on the current policy $u(t)=-K_ix(t)$.
	\item At time $t+T$, collect and compute the data set $(X,Y)$, which are defined in Eq. \ref{Eq:X} and Eq. \ref{Eq:Y}.
	\item Update vector $w$  by using the batch least square method (Eq. \eqref{Eq:LS}).
\end{itemize}

\subsubsection*{Control policy update}
\begin{itemize}
	\item Transform vector $w$ into the kernel matrix $P$ using the Kronecker transformation.
	\item Update the insulin policy by solving the LMI in Theorem \ref{Theorem:Policy_Update}.
\end{itemize}

\subsection{Simulation Setup}
A simulation study of the proposed robust RL controller was conducted on a glucose kinetics model, which can be described by \cite{Bergman1979, Hovorka2004, Wilinska2005a, Mosching2016}:

\begin{equation}
\frac{dD_1(t)}{dt}=A_GD(t) - \frac{D_1(t)}{\tau_D}
\label{}
\end{equation}
\begin{equation}
\frac{dD_2(t)}{dt}=\frac{D_1(t)}{\tau_D} - \frac{D_2(t)}{\tau_D}
\label{}
\end{equation}
\begin{equation}
\frac{dg(t)}{dt} = -p_1g(t) -\chi(t)g(t) + \frac{D_2(t)}{\tau_D} + w(t)
\label{Eq:BG_evolution}
\end{equation}
\begin{equation}
\frac{d\chi(t)}{dt} = -p_2\chi(t) + p_3 V \left(i(t)-i_b(t)\right)
\label{}
\end{equation} 
In this model, variable descriptions and parameter values can be found in Table 1 and Table 2. Variable $w(t)$ is the process noise. The measured blood glucose value is affected by a random noise $v(t)$:
\begin{equation}
\hat{g}(t) = g(t) + v(t)
\label{measured_noise}
\end{equation}
The inputs of the model are the amount of CHO intake $D$ and the insulin concentration $i$. The value of $i(t) - i_b(t)$ must be non-negative:
\begin{equation}
i(t) - i_b(t) \geq 0
\label{Eq:insulin_sturation}
\end{equation}
\begin{table}[h] 
	\centering
	\caption{Glucose kinetics model's parameters \cite{Bergman1979,Hovorka2004}} 
	\label{table:parameters}
	\begin{tabular}{|c|c|c|}
		\hline
		\multicolumn{2}{|c|}{Description} & Value \\ 
		\hline
		$p_1$ & Glucose effectiveness & 0.2 $ \mathrm{min}^{-1}$\\
		\hline
		$p_2$ & Insulin sensitivity& 0.028 $ \mathrm{min}^{-1}$\\
		\hline
		$p_3$ & Insulin rate of clearance& $10^{-4} \ \mathrm{min}^{-1}$\\
		\hline
		$A_G$ & Carbohydrate bio-availability& $0.8 \ \mathrm{min}^{-1}$\\
		\hline
		$\tau_D$ & Glucose absorption constant& $10$ min\\
		\hline
		$V$ & Plasma volume& $2730$ g\\
		\hline
		$i_b(t)$ & Initial basal rate&  $7.326 \ \mu \mathrm{IU}/\mathrm{(ml.min)}$ \\
		\hline
	\end{tabular}
\end{table}

\begin{table}[h]
	\centering
	\caption{Variables of the glucose kinetics model}
	\label{table:variables}
	\begin{tabular}{|c|c|c|}
		\hline
		\multicolumn{2}{|c|}{Description} & Unit \\ 
		\hline
		$D$ & Amount of carbohydrate intake & mmol/min\\
		\hline 
		$D_1$ &  Glucose in compartment 1 &  mmol\\ 
		\hline 
		$D_2$ &  Glucose in compartment 2&  mmol\\ 
		\hline 
		$g(t)$ &  Plasma glucose concentration &  mmol/l\\ 
		\hline 
		$\chi(t)$ &  Interstitial insulin activity&  $\mathrm{min}^{-1}$ \\
		\hline 
		$i(t)$ &  Plasma insulin concentration &  $\mu \mathrm{IU}\mathrm{/ml}$ \\
		\hline
		
	\end{tabular}
\end{table}

\section{Results}
In order to evaluate the performance of the robust RL controller, we implemented the controller on the glucose kinetics model as described in the previous section under a daily scenario of patients with type-1 diabetes. In order to make the scenario realistic, three different levels of uncertainties were used in the model. Uncertainties include process noise ($w(t)$) and measurement noise ($v(t)$). It is assumed that the noises are Gaussian distributions with standard deviations for each case as shown in Table \ref{table:Uncertainty_case}.

\begin{table}
	\centering
	\caption{Standard deviations of process and measurement noises}
	\label{table:Uncertainty_case}
	\begin{tabular}{|c|c|c|}
		\hline 
		Uncertainty case & Process noise (w(t)) & Measurement noise (v(t))\\ 
		\hline 
		1 & 0 & 0 \\ 
		\hline 
		2 & 0 & 0.002  \\ 
		\hline 
		3 & 0.1 & 0.1  \\ 
		\hline
		4 & 0.1 & 1  \\ 
		\hline 
	\end{tabular} 
\end{table}


\subsection{Without Meal Intake}
This part describes the simulation results during the fasting period (without meal intake). The purpose of the simulation is to compare the performances of the robust RL algorithm with the conventional optimal RL algorithm in the nominal condition (uncertainty case 1). The initial blood glucose for both scenario was set at 290 mg/dL and the target blood glucose is 90 mg/dL. The initial policy at the beginning of the simulation was chosen as follows: 
\begin{equation}
{u}(t)=-K_0x(t)= -0.27g(t) + 266.00\chi(t)
\end{equation}

\begin{figure}[h]
	\centering
	\includegraphics[width=0.8\linewidth]{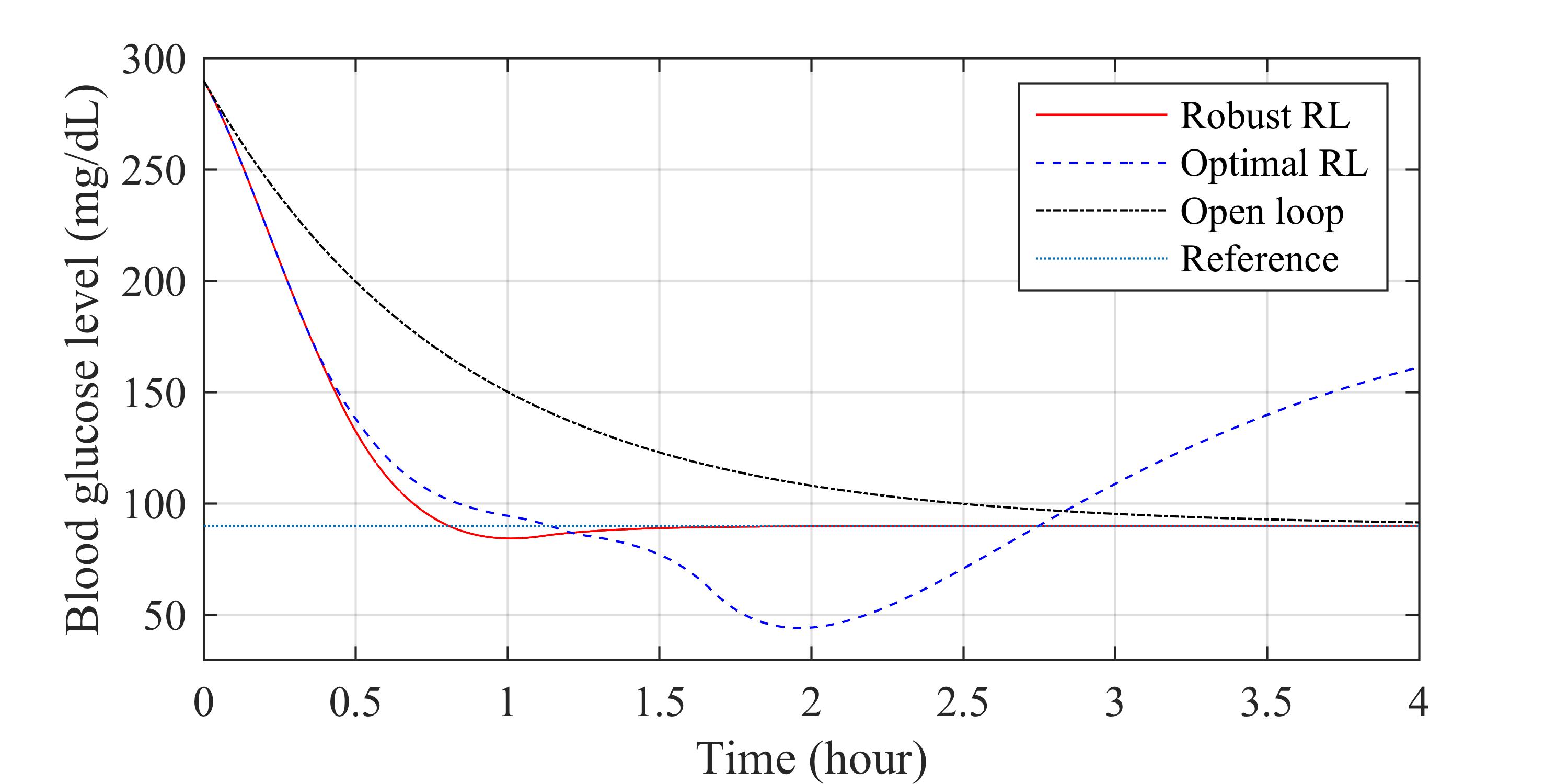}
	\caption{Comparison of blood glucose responses in nominal case without meal intake}
	\label{fig:BG_nominal_4h}
\end{figure}

\begin{figure}[h]
	\centering
	\includegraphics[width=0.8\linewidth]{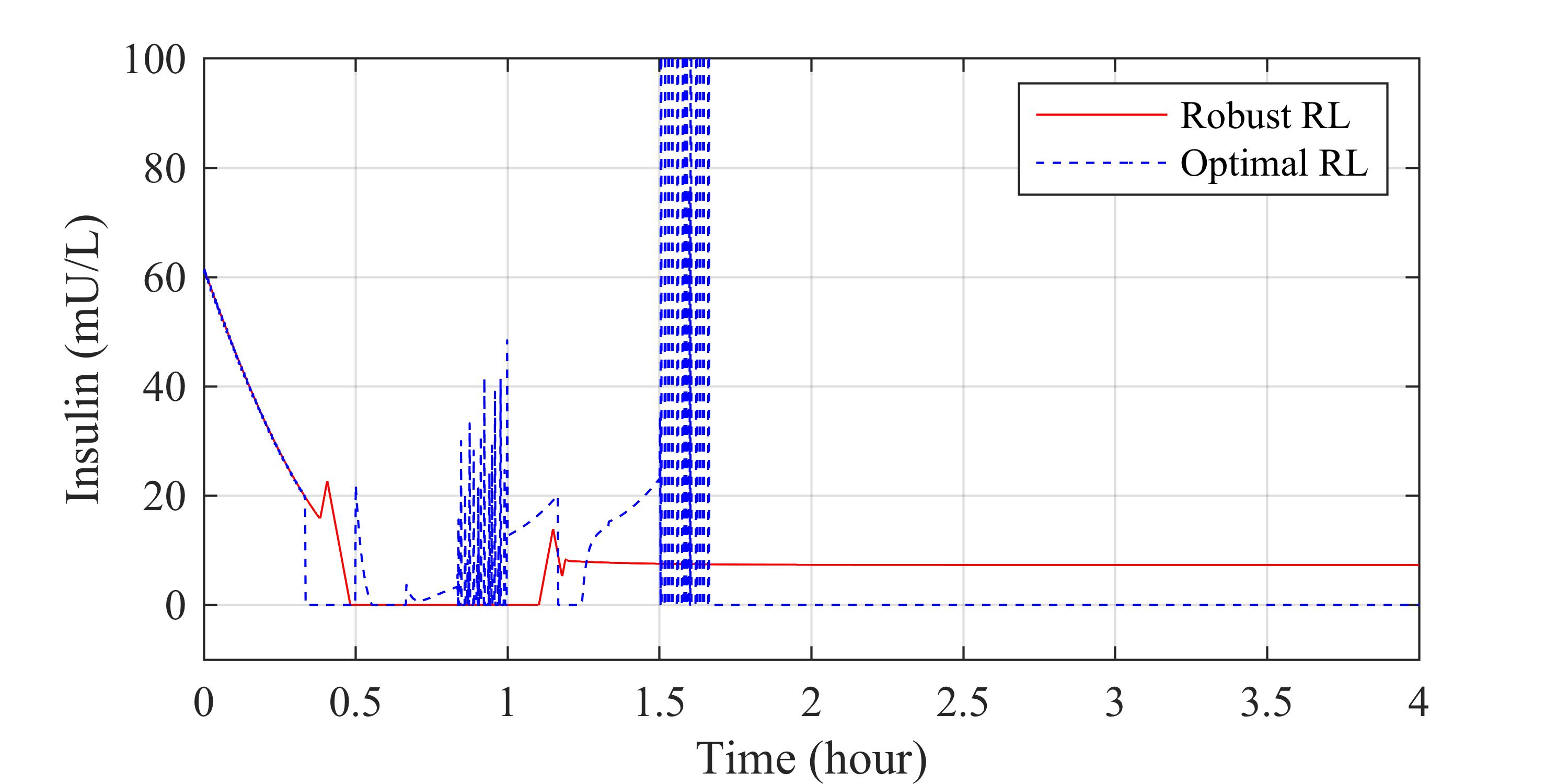}
	\caption{Comparison of insulin concentration in nominal case without meal intake}
	\label{fig:insulin_nominal_4h}
\end{figure}

Fig. \ref{fig:BG_nominal_4h} shows the comparison in blood glucose level between the robust RL and the optimal RL algorithm in the nominal condition. From the results, it can be seen that the robust RL successfully reduces the blood glucose level while the optimal RL becomes unstable when the blood glucose approaches the desired value. The instablity of the optimal RL in this case can be explained by the nonlinearity of the system (due to the coupling term $\chi(t)g(t)$ in Eq. \ref{Eq:BG_evolution}), the saturation of the insulin concentration (Eq. \ref{Eq:insulin_sturation}), and the lack of perturbed data when the blood glucose approaches the steady state value. The insulin concentration during the simulation can be found in Fig. \ref{fig:insulin_nominal_4h}. In this figure, the dotted blue line indicates the unstable insulin profile.

\begin{figure}[h]
	\centering
	\includegraphics[width=0.8\linewidth]{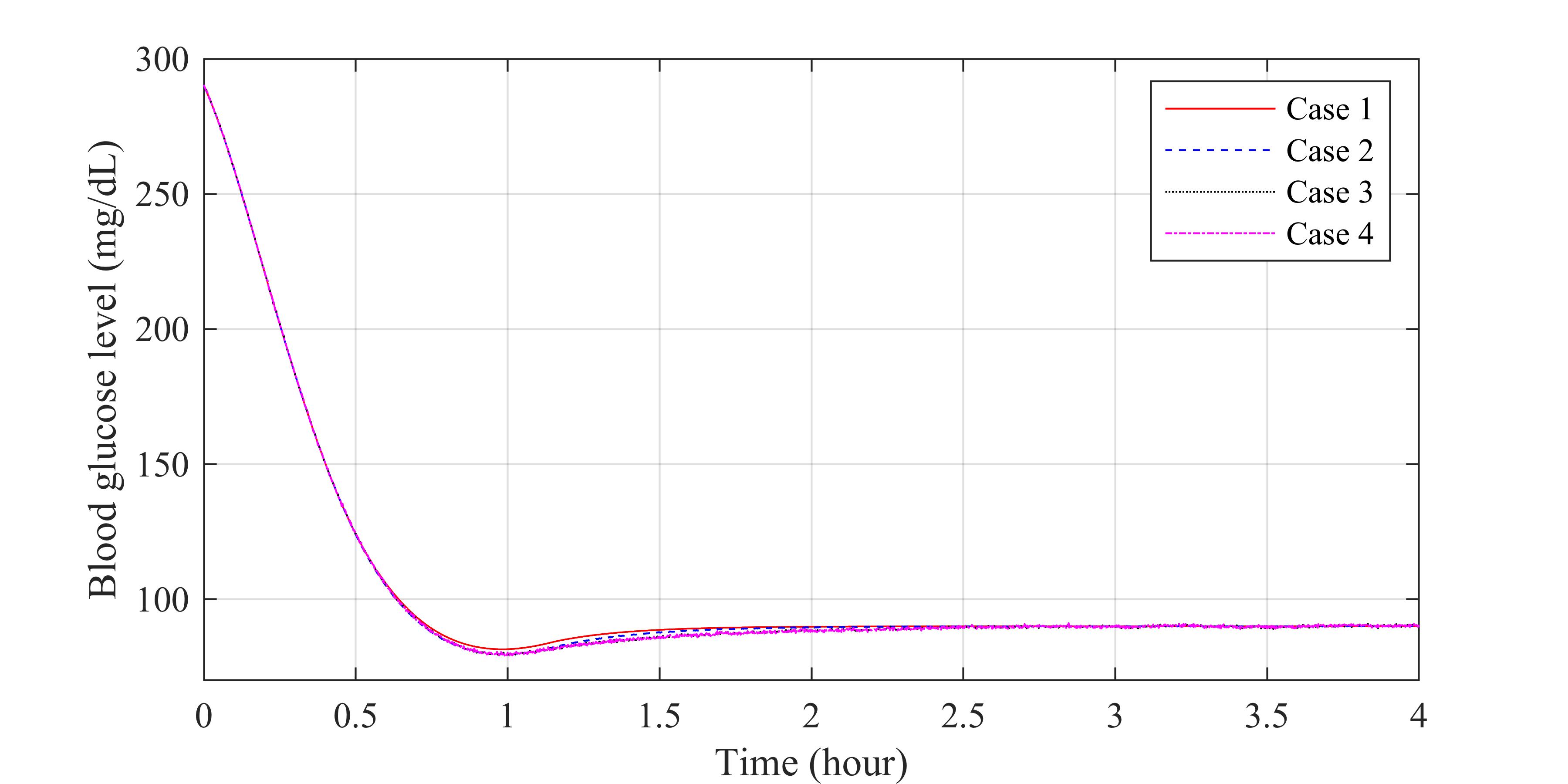}
	\caption{Comparison of blood glucose responses in uncertain cases without meal intake}
	\label{fig:BG_uncertain_4h}
\end{figure}

Fig. \ref{fig:BG_uncertain_4h} shows the blood glucose responses from the robust RL in different uncertain conditions without meal intake. The results show similar and stable responses in all the uncertain conditions with settling time to the desired blood glucose level of approximately 45 minutes. 
The insulin concentration and the update of controller gains can be found in Fig. \ref{fig:insulin_uncertain_4h} and Fig. \ref{fig:K}.

\begin{figure}[h]
		\centering
	\includegraphics[width=0.8\linewidth]{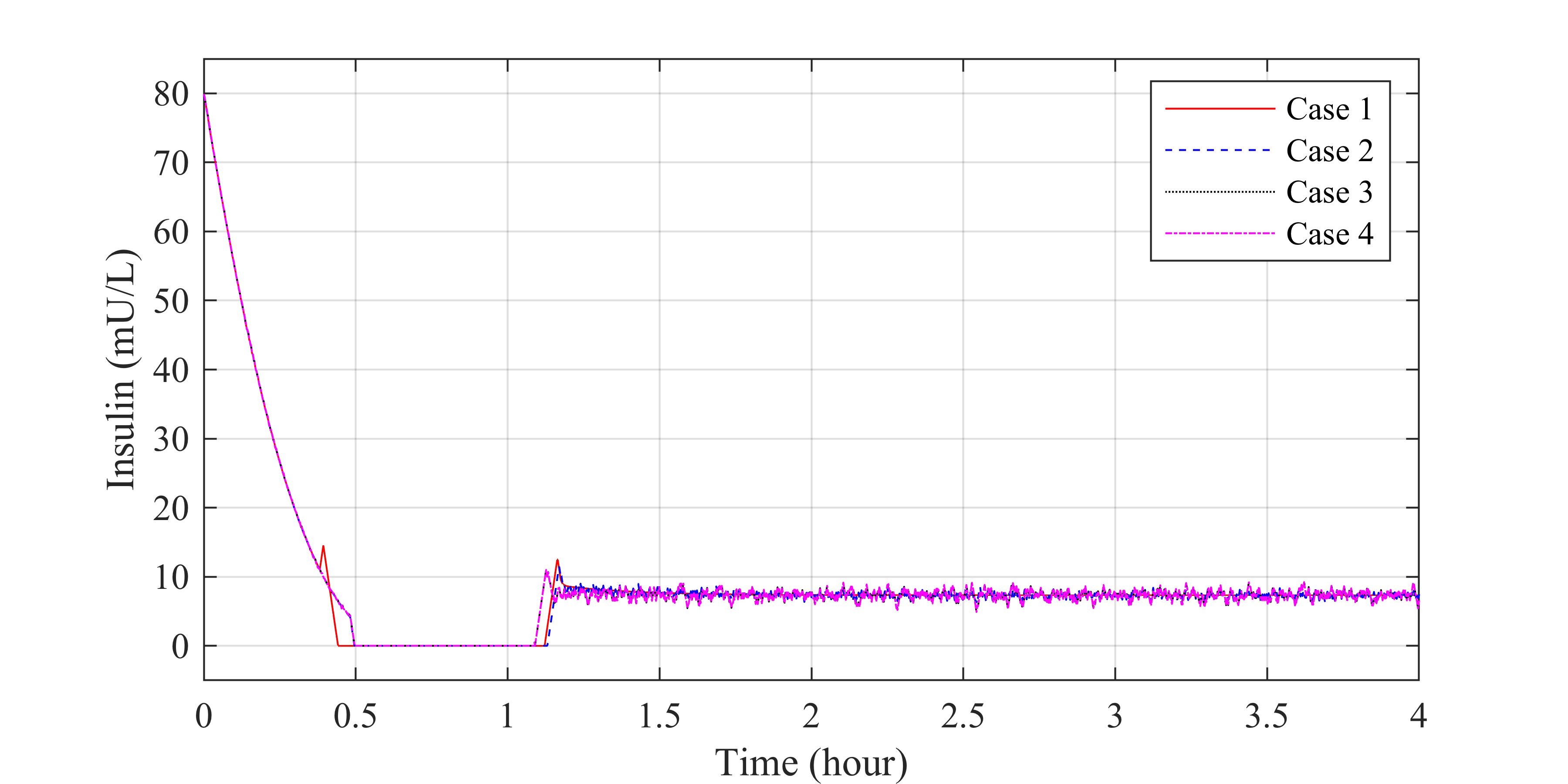}
	\caption{Insulin concentration in uncertain cases without meal intake}
	\label{fig:insulin_uncertain_4h}
\end{figure}

\begin{figure}[h]
		\centering
	\includegraphics[width=0.8\linewidth]{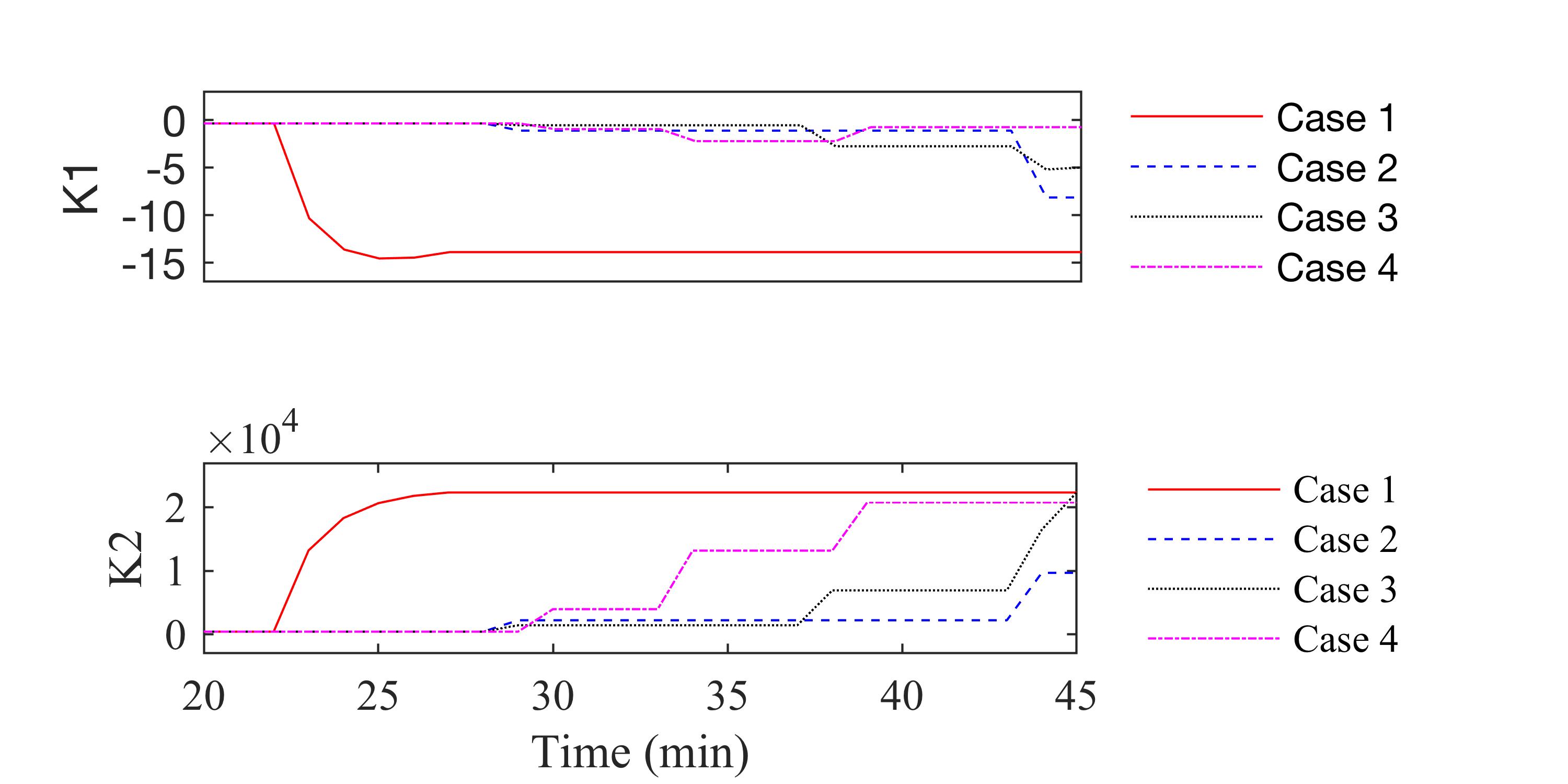}
	\caption{Update of controller gains during the learning process ($K1$ and $K2$ represent the first and second element of the controller gain vector $K$)}
	\label{fig:K}
\end{figure}


\subsection{With Meal Intake}
In this part, the performance of the robust RL controller was tested where the system is subjected to meal intakes with the carbohydrate profile as shown in Fig. \ref{fig:CHO}.
\begin{figure}[h]
		\centering
	\includegraphics[scale=0.1]{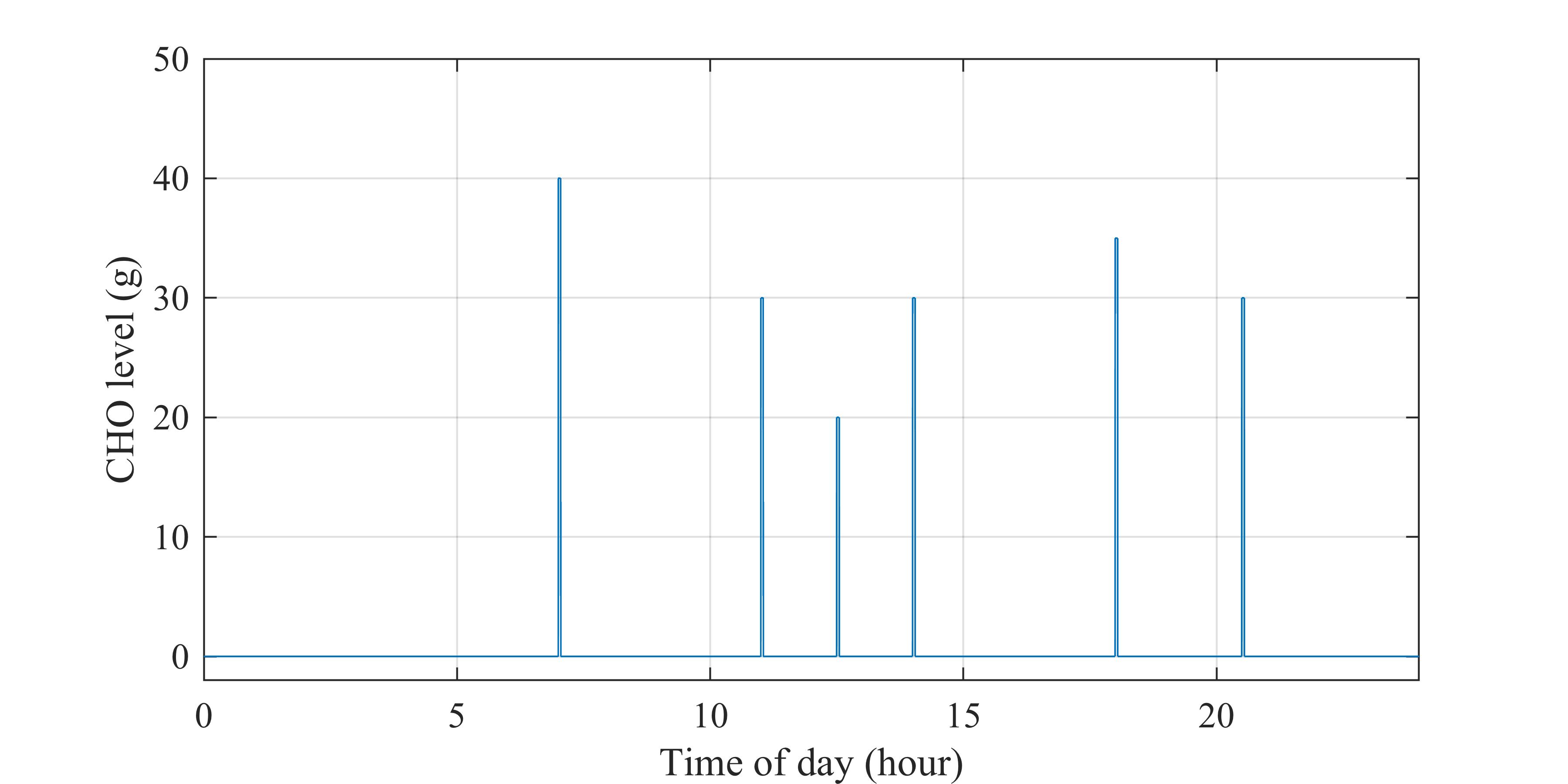}
	\caption{Carbohydrate intake per meal}
	\label{fig:CHO}
\end{figure}

During the simulation period with meal intakes, blood glucose responses throughout the day of the robust RL control systems under four uncertain cases
are shown in Fig. \ref{fig:BG_uncertain_1day}. The insulin concentration during the process can also be found in Fig. \ref{fig:insulin_uncertain_1day}. The results show that the controller provides the most aggressive action under case 1 (no uncertainty) and the least aggressive action under case 4 (with highest level of measurement and process noises). This leads to the largest and smallest reduction of postprandial blood glucose in case 1 and case 4, respectively. Most importantly, the robust RL algorithm kept the system in stable condition and there is no hypoglycemia event during the simulation for all four cases under different level of uncertainties.

\begin{figure}[h]
	\centering
	\includegraphics[width=0.8\linewidth]{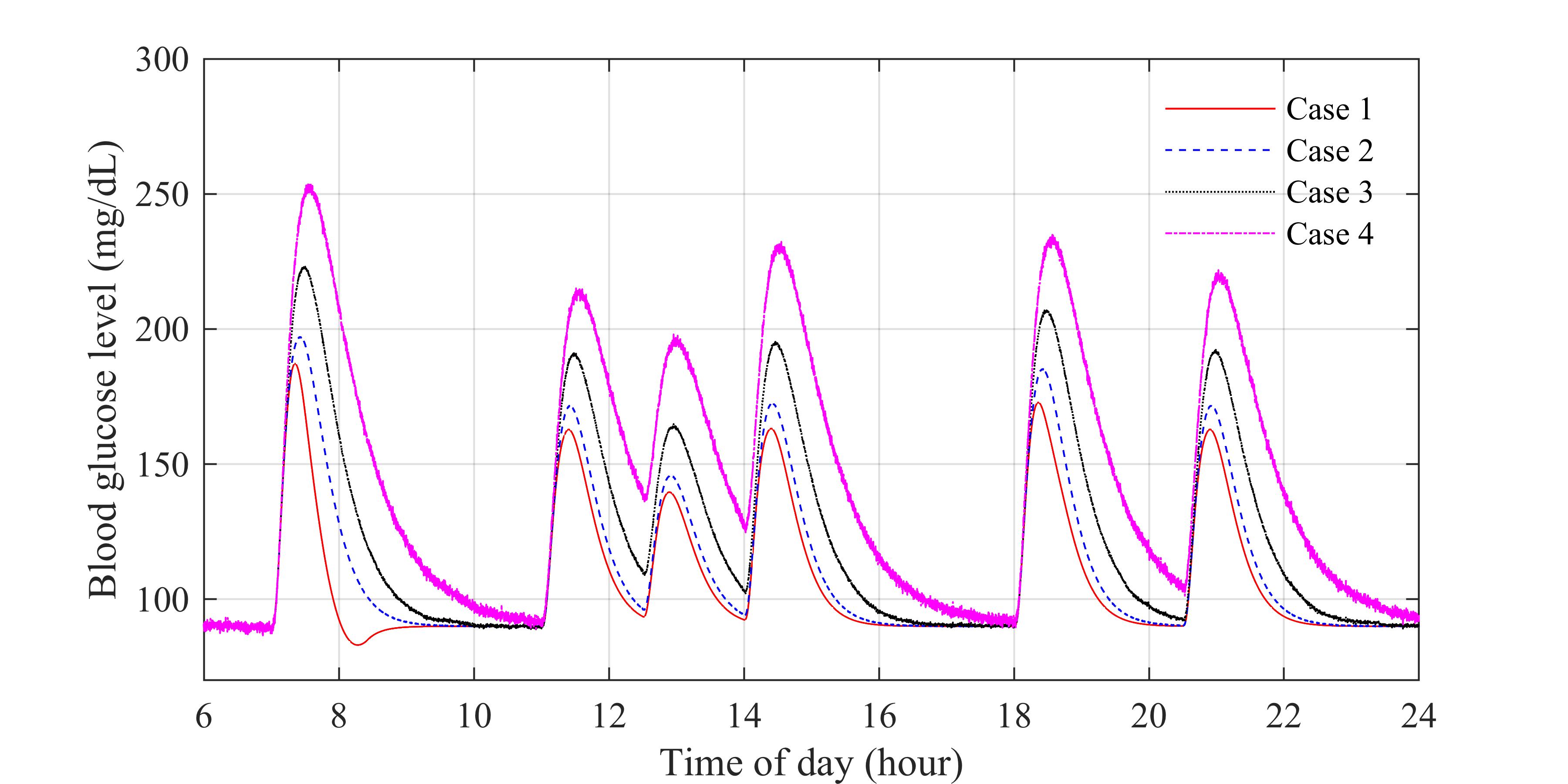}
	\caption{Blood glucose responses in simulation with meals}
	\label{fig:BG_uncertain_1day}
\end{figure}

\begin{figure}[h]
	\centering
	\includegraphics[width=0.8\linewidth]{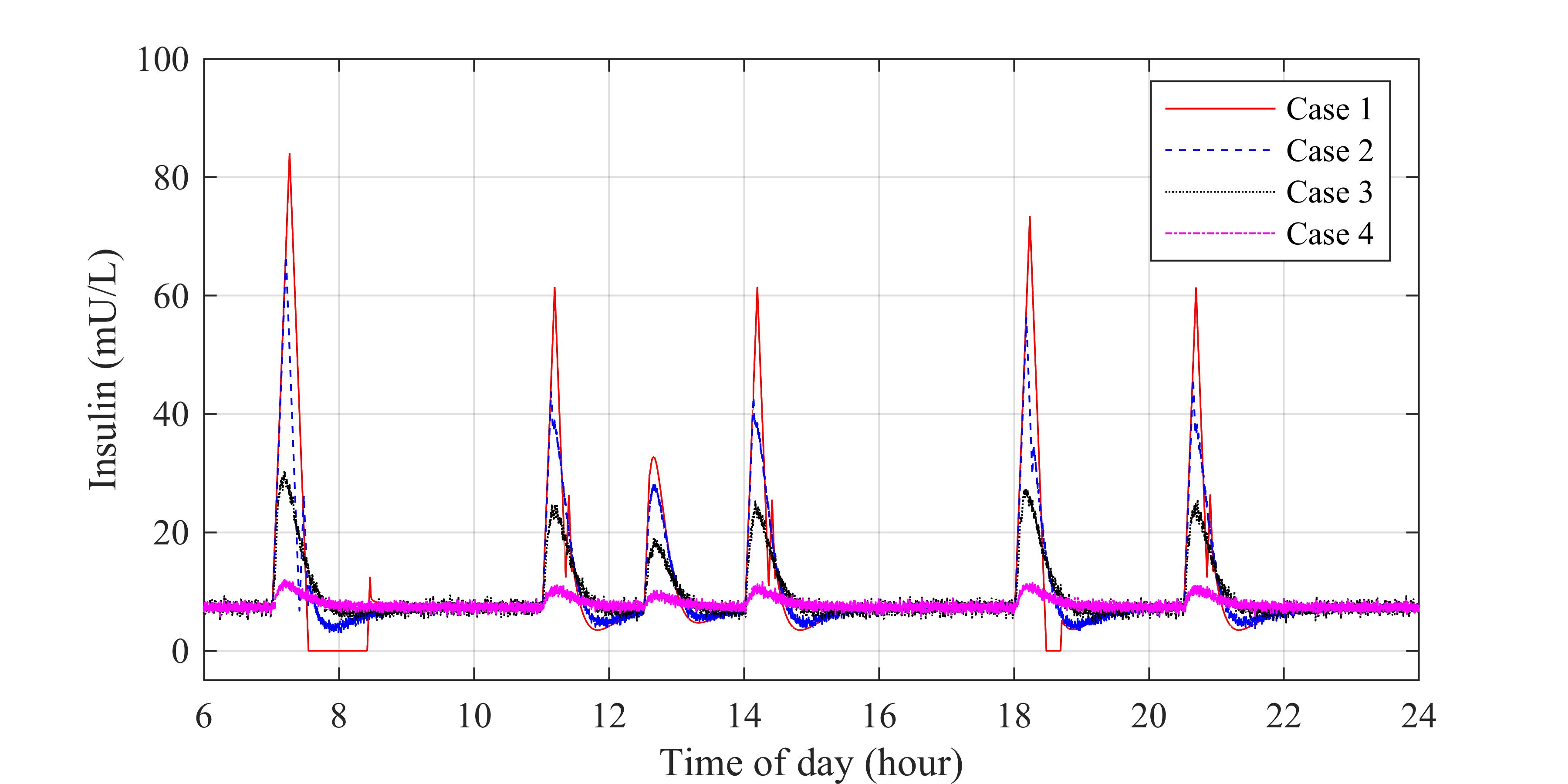}
	\caption{Insulin concentration in simulation with meals}
	\label{fig:insulin_uncertain_1day}
\end{figure}

\section{Conclusion}
The paper proposes a robust reinforcement learning algorithm for dynamical systems with uncertainties. The uncertainties can be approximated by the critic and represented in the value function. LMI techniques were used to improve the controller gain. The algorithm was simulated on a blood glucose model for patients with type-1 diabetes. The objective of the simulation is to control and maintain a healthy blood glucose level. The comparison between the robust RL algorithm and the optimal RL algorithm shows a significant improvement in the robustness of the proposed algorithm. Simulation results show that the algorithm successfully regulated the blood glucose and kept the system stable under different levels of uncertainty.


%

\section*{Acknowledgment}
The research was funded by Tromsø Research Foundation.

\clearpage
\bibliographystyle{IEEEtran}
\bibliography{library}

\end{document}